\documentclass[a4paper,11]{article}

\usepackage{format-int}
\usepackage{graphicx}
\usepackage{epsf}

\newcommand{\apref}[1]{Appendix~\ref{sec:#1}}

\title{The {\em Umwelt\/} of an Embodied Agent -- \\
A Measure-Theoretic Definition\footnote{Preprint of an article in \emph{Theory in Biosciences}, 134 no.\ 3, 2015. doi:10.1007/s12064-015-0217-3}} 

\author{Nihat Ay${}^{1,2,3}$  \and Wolfgang L\"ohr${}^{4}$ }

\begin{document}

\maketitle

\begin{center}
${}^{1}$Max Planck Institute for Mathematics in the Sciences, Inselstrasse 22, 04103 Leipzig, Germany \\
${}^{2}$Faculty of Mathematics and Computer Science, University of Leipzig, \\ PF 100920, 04009 Leipzig, Germany \\
${}^{3}$Santa Fe Institute, 1399 Hyde Park Road, Santa Fe, New Mexico 87501, USA \\
${}^{4}$Universit\"at Duisburg-Essen, Thea-Leymann-Strasse 9, 45117 Essen, Germany \\ 
\end{center}

\begin{abstract} 
We consider a general model of the sensorimotor loop of an agent interacting with the world. This formalises Uexk{\"u}ll's notion of a {\em function-circle\/}. 
Here, we assume a particular causal structure, mechanistically described in terms of Markov kernels. In this generality, we 
define two $\sigma$-algebras of events in the world that describe two respective perspectives: (1) the perspective of an external observer, (2) the intrinsic perspective of the agent. Not all aspects of the world, seen from the external perspective, are accessible to the agent. This is expressed 
by the fact that the second $\sigma$-algebra is a subalgebra of the first one. 
We propose the smaller one as formalisation of Uexk{\"u}ll's {\em Umwelt\/} concept. 
We show that, under continuity and compactness assumptions, 
the global dynamics of the world can be simplified without changing the internal process. 
This simplification can serve as a minimal world model that the system must have in order to be consistent with the internal process.   \\
 

{
\small
\tableofcontents
}
\end{abstract}


\newpage

\section{Introduction: the intrinsic view of embodied agents}

\subsection{Uexk{\"u}ll's {\em function-circle\/} and the sensorimotor loop}
A key observation based on many case studies within the field of embodied intelligence 
implies that quite simple control mechanisms can 
lead to very complex behaviours \cite{PfeiferBongard}. 
This gap between simplicity and complexity related to the same thing, the agent's behaviour, is the result of two different frames of 
reference in the description. Here, the intrinsic view of the agent, which provides the basis for its control, 
can greatly differ from the (extrinsic) view of an external observer.
This important understanding is not new. In the first half of the last century, Uexk{\"u}ll has conceptualised this understanding by his notion of {\em Umwelt\/}, which summarises all 
aspects of the world that have an effect on the agent and can be affected by the agent (see \cite{Uexkull}). Furthermore, he has convincingly   
exemplified this notion in terms of many 
biological case studies. These case studies are presented in his book \cite{Uexkull34}, supplemented by insightful illustrations (see in Figure \ref{biene}, as an example, the {\em Umwelt\/} of a bee).  \\

\begin{figure}[h]
\begin{center}        
               \includegraphics[width=14.5cm]{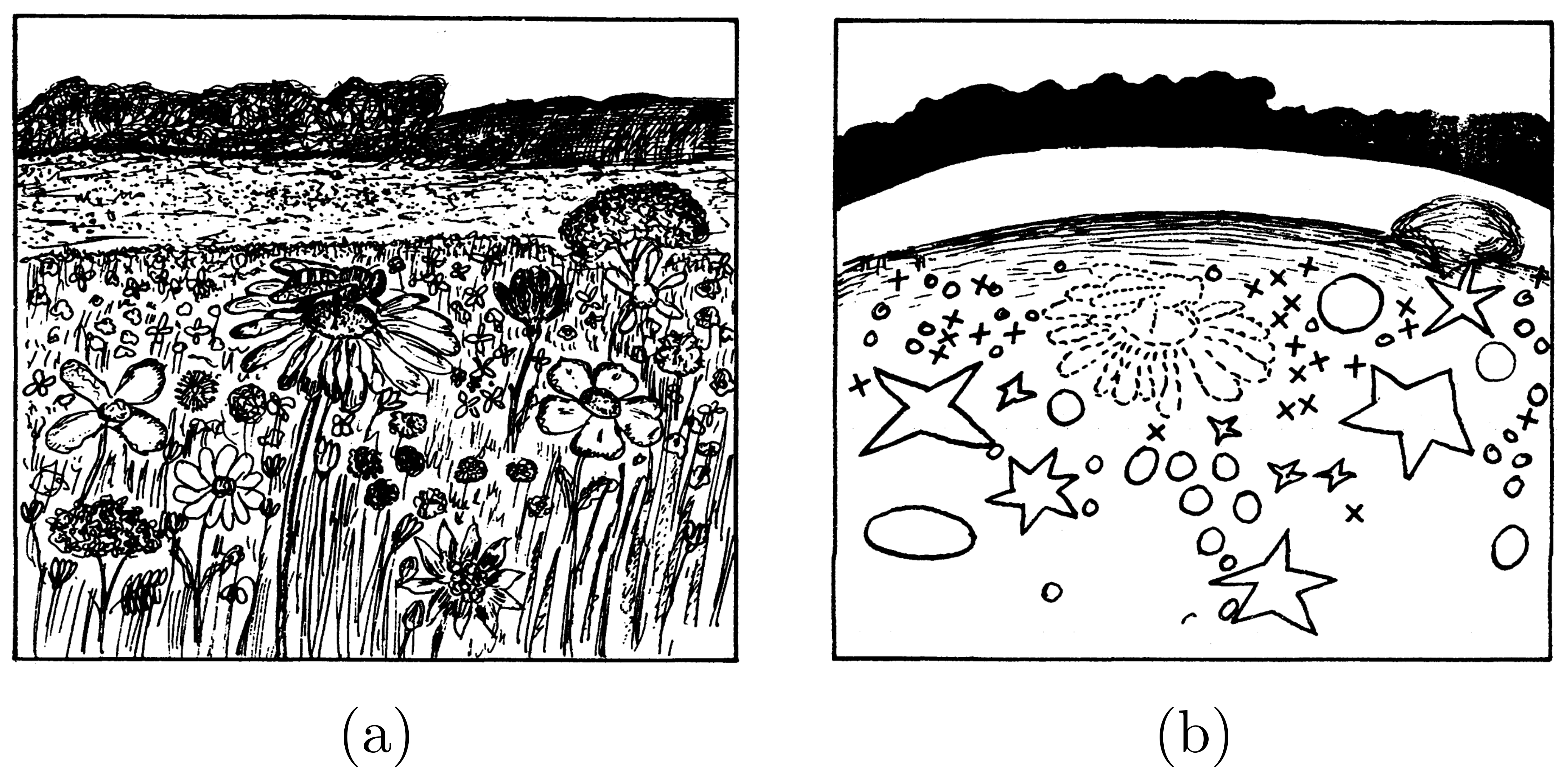}
                \vspace{-2mm}
\end{center}
\caption{The {\em Umwelt\/} of a bee as illustrated in \cite{Uexkull34}. (a) The environment of a bee how we perceive it as an external observer. (b) The same bee perceives only particular aspects of the same world, which 
constitute its {\em Umwelt\/}.}
\label{biene}
\end{figure}

Uexk{\"u}ll has developed his {\em Umwelt\/} concept based on the notion of a {\em function-circle\/} ({\em Funktions-kreis\/}, see Figure \ref{funktionskreis} (a)). It graphically represents 
the causal interaction of an animal with its surroundings. Nowadays this circle is known as the {\em sensorimotor loop\/} and it plays an important role 
within the field of embodied cognition (see Figure \ref{funktionskreis} (b)). However, its interpretation has not changed, so that Uexk{\"u}ll's description of the {\em function-circle\/}
perfectly applies to the sensorimotor loop:  

\begin{quote}
``Every animal is a subject, which, in virtue of the structure peculiar to it, selects stimuli from the general influences of the outer world, and to these it responds in a certain way. These responses, 
in their turn, consist of certain effects on the outer world, and these again influence the stimuli. In this way there arises a self-contained periodic cycle, which we may call the function-circle of the animal.'' 
(\cite{Uexkull26}, page 128)  
\end{quote}

\begin{figure}[h]
\begin{center}        
               \includegraphics[width=14.5cm]{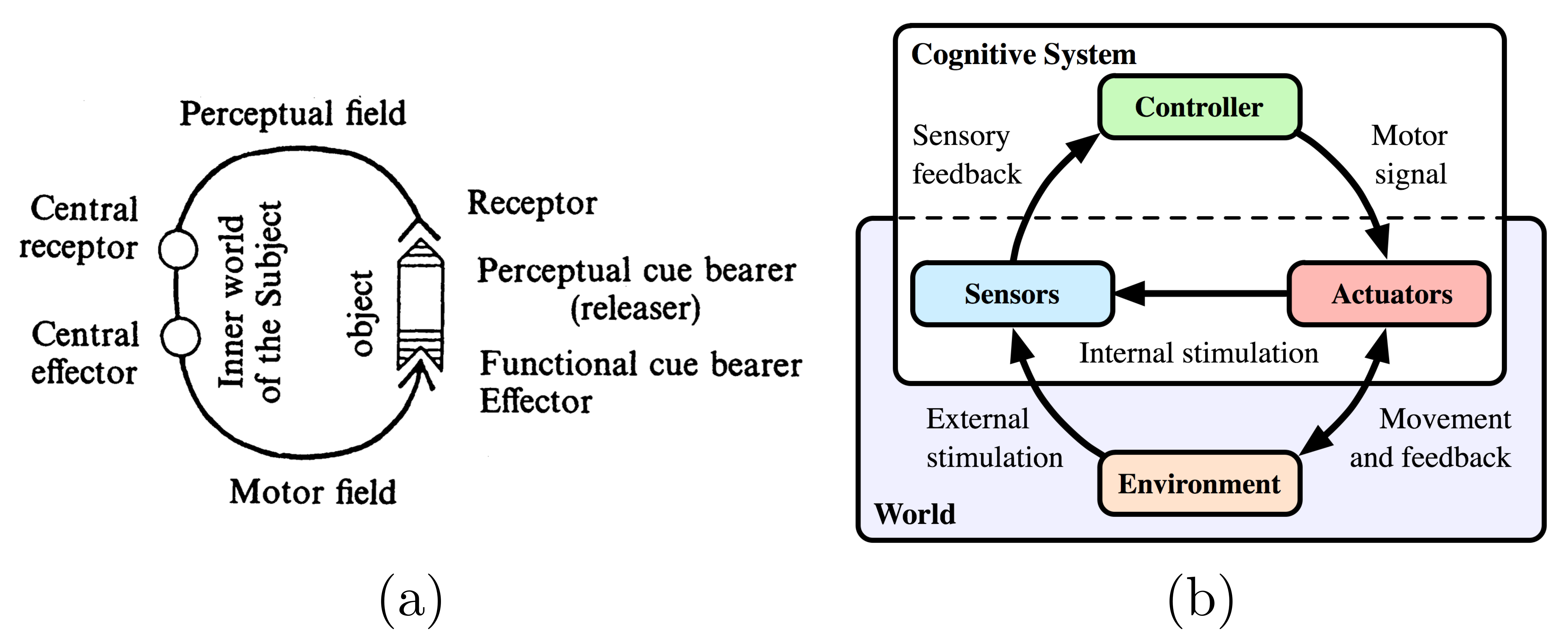}
               \vspace{-5mm}
\end{center}
\caption{(a) Uexk{\"u}ll's {\em function-circle\/} ({\em Funktionskreis\/}) \cite{Uexkull34}, (b) the sensorimotor loop from the field of embodied cognition \cite{AyZahedi}.} 
\label{funktionskreis}
\end{figure}
In this section, we introduce the {\em sensorimotor loop\/} in terms of a causal diagram \cite{Pearl} of the involved processes, 
which describes the interaction of the agent with the world \cite{TishbyPolani, ZahediAyDer}. 
In addition to the causal structure of this interaction, as shown in Figure \ref{sensomot}, we need to formalise the mechanisms that generate the individual processes.  
A very general way of formalising mechanisms is provided by the notion of a Markov kernel, which is also used in information theory for a mathematical description of a channel.
As an example, let us consider the sensor mechanism denoted by $\beta$. Given a state $w$ of the world, the sensor assumes a state $s$, which can be subject to some noise. Therefore, the 
sensor output is best described as a probability distribution over the sensor states, that is $\beta(w;ds)$. This reads as the probability that the sensor assumes a state $s$ in the infinitesimal  
set $ds$. Equivalently, we can consider the probability $\beta(w;B)$ that the sensor assumes a state $s$ in the set $B$ given the world state $w$. Clearly, we get this probability of $B$ by integrating 
the probabilities of the infinitesimal sets $ds$ contained in $B$: $\beta(w;B) = \int_B \beta(w; ds)$. In addition to the sensor mechanism $\beta$ we have the actuator mechanism $\pi$, referred to as the agent's {\em policy\/}. Finally, 
the mechanisms $\alpha$ and $\varphi$ describe the dynamics of the world and the agent, respectively. Note that the probabilistic description of the mechanisms does not mean that we exclude deterministic mechanisms. 
For instance, one might want to assume that the dynamics of the world is deterministic. Together with an initial distribution of $W_0,S_0,C_0,A_0$, the Markov kernels $\alpha$, $\beta$, $\varphi$, and $\pi$, 
the mechanisms of the sensorimotor loop, specify the distribution of the overall process $W_n, S_n, C_n, A_n$, $n\in {\Bbb N}$, 
consisting of the individual processes of the world (W), the sensors (S), the agent (C), and the actuators (A).   

\newcommand{\actkern}{{\qquad\qquad\raisebox{1ex}{$\alpha$}}}
\newcommand{\intkern}{{\qquad\qquad\raisebox{1ex}{$\varphi$}}}
\newcommand{\senkern}{{\textstyle\beta}}
\newcommand{\polkern}{{\textstyle\pi}}

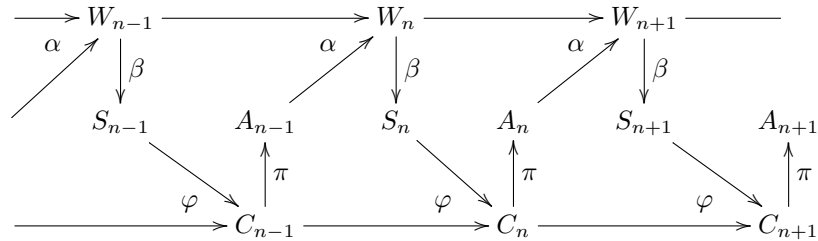
\begin{figure}[h]
\begin{center}        
\[ \xymatrix{
	\ar[r]&	W_{n-1}\ar[rr]\ar[d]^\senkern & & W_n\ar[rr]\ar[d]^\senkern & & W_{n+1}\ar@{-}[r]\ar[d]^\senkern &\\
	\ar[ur]^\actkern& S_{n-1}\ar[dr] & A_{n-1} \ar[ur]^\actkern & S_n\ar[dr] & A_n\ar[ur]^\actkern & S_{n+1}\ar[dr]&A_{n+1}\\
	\ar[rr]^\intkern  &&{C_{n-1}} \ar[rr]^\intkern \ar[u]_\polkern & & {C_n} \ar[rr]^\intkern \ar[u]_\polkern & & {C_{n+1}}\ar[u]_\polkern
} \]
\end{center}
\caption{The causal diagram of the sensorimotor loop. In each instant of time the agent (C) takes a measurement from the world (W) 
through its sensors (S) and affects the world through its actuators (A). 
}
\label{sensomot}
\end{figure}

The causal model of the sensorimotor loop will 
allow us to formalise what we mean by intrinsic and extrinsic frames of reference in terms of $\sigma$-algebras.
These are basic mathematical objects from measure theory that naturally describe a set of observables which is assigned 
to an observer.  
Already at this very general level, a description of the gap between the intrinsic and extrinsic perspective is possible. 


%

\subsection{Uexk{\"u}ll's {\em Umwelt\/} in terms of a $\sigma$-algebra}
With this structure at hand, we can formalise Uexk{\"u}ll's agent centric {\em Umwelt\/}. Note that the world considered in the sensorimotor loop is meant to be 
the world as it can be seen from the perspective of an external observer, which is also referred to as {\em outer world\/} by Uexk{\"u}ll \cite{Uexkull26}. 
This perspective is clearly not accessible to the agent. 
The agent has its own intrinsic view at this world, which implies an agent specific world, consisting only of those objects in the world that the agent can perceive and affect. \\

Now let us be a bit more precise. Assume that we have two world states $w$ and $w'$ that are distinguishable from the perspective of an external observer. The agent's sensorimotor apparatus, however, may not 
be rich enough for this distinction. 
The agent would then perceive $w$ and $w'$ as being the same 
world state. By this identification, the original set of world states is partitioned into classes that represent the states of the agent's world, its {\em Umwelt\/}. 
This partition is illustrated in Figure \ref{internalstates}. \\

\begin{figure}[h]
\begin{center}        
             \includegraphics[width=6cm]{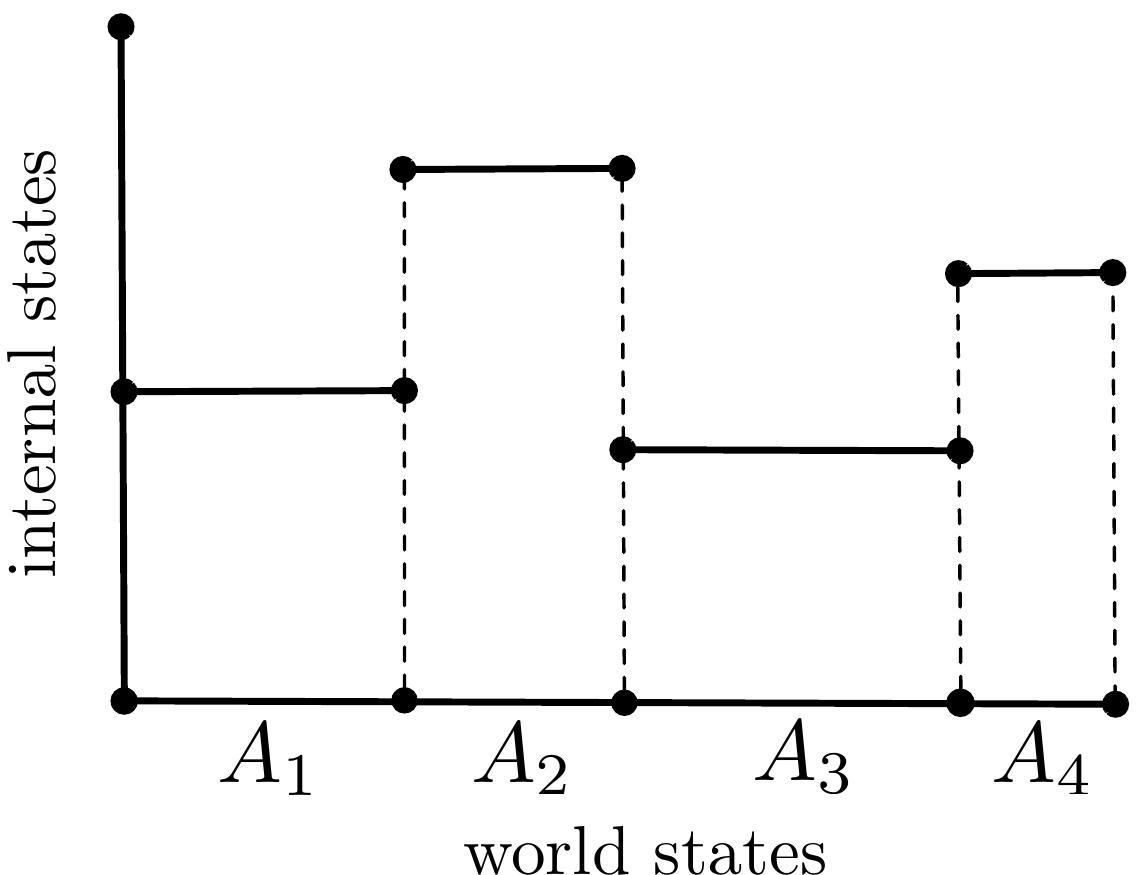} 
               \vspace{-1mm}
\end{center}
\caption{Clustering of world states.}
\label{internalstates}
\end{figure}
Let us first assume that we have only finitely many of these classes, say $A_1, A_2, \dots,A_n$, which represent the internal states of the agent. 
Given a world state $w$, the agent will assume one of these classes as internal state, the one that contains $w$, which we denote by $A(w)$.
Now consider an arbitrary subset $A$ of world states.
We call $A$ a {\em distincion\/} (that the agent can make with its internals states) if the following holds for all world states $w$: Knowing the internal state $A(w)$ is sufficient 
to decide whether or not $w$ is contained in $A$.  
Clearly, the individual classes themselves are distinctions. But these are not the only ones. For example, $A = A_1 \cup A_2$ is also a distinction. 
To see this we have to consider three cases: (1) $w \in A_1$: then $w \in A$, (2) $w \in A_2$: then $w \in A$, (3) $w \in A_i$ for some $i > 2$: then $w \notin A$. Thus, if we know to which class $w$ belongs then we know whether or not 
$w$ is contained in $A$. More generally, we have the following set ${\mathcal A}$ of distinctions that the agent can make with the internal states $A_1,\dots,A_n$:
\begin{equation} \label{endlichealgebra}
      {\mathcal A} 
       \; = \; \left\{ A_{i_1} \cup A_{i_2} \cup \dots \cup A_{i_k} \; : \; 1 \leq i_1 < i_2 < \dots < A_{i_k} \leq n \right\}.
\end{equation}
This set is closed under natural operations. 
Clearly, if $A$ is a distinction, then the complement $A^c$ of $A$ is also a distinction. Furthermore, unions and intersections of distinctions are also distinctions.     
Having such a set ${\mathcal A}$ of distinctions, we can recover the class to which a world state $w$ belongs by 
\begin{equation} \label{class}
         A(\omega) \, = \, {[w]}_{\mathcal A} \, := \, \bigcap_{A \in {\mathcal A} \atop w \in A} A .  
\end{equation}
This correspondence between partitions into classes and sets of distinctions is one-to-one in the finite case. However, if we drop that assumption 
the correspondence does not hold anymore, and we have to work with sets of distinctions in the first place, as they encode more information than the corresponding partitions.   
Extending our reasoning to possibly infinitely many distinctions, we have to assume that the set ${\mathcal A}$ of distinctions an agent can make with its sensorimotor apparatus satisfies the 
conditions of a $\sigma$-algebra, that is:
\begin{enumerate}
\item $\emptyset \in {\mathcal A}$, 
\item $A \in {\mathcal A} \quad \Rightarrow \quad A^c \in {\mathcal A}$, 
\item $A_1,A_2,\dots \in {\mathcal A} \quad \Rightarrow \quad \bigcup_{i=1}^\infty A_i \in {\mathcal A}$.
\end{enumerate}
Clearly, as a special case, the set (\ref{endlichealgebra}) forms a $\sigma$-algebra. Given an arbitrary $\sigma$-algebra ${\mathcal A}$, we can use formula (\ref{class}) in order to 
define the classes that it generates. As already mentioned, this set of classes alone contains less information than the set ${\mathcal A}$ of distinctions. 
This is the reason why we consider in this article state sets {\em together\/} with $\sigma$-algebras on these state sets, leading to the notion of so-called measurable spaces.  \\

Note that we do not address the problem of how the sensorimotor apparatus of the agent
might instantiate a set of distinctions.  
In this regard, we want to highlight the following problem of information integration. Let us associate with a distinction $A$ a sensor of the agent that is active if the world state $w$ is in $A$, and inactive if it is not. 
Then consider two of such sensors corresponding to the distinctions $A_1$ and $A_2$. In principal, knowing that $w$ is in $A_1$, through the first sensor, 
and knowing that $w$ is also in $A_2$, through the second sensor, implies that $w \in A_1 \cap A_2$. 
But this implication is a purely logical one. It is not clear whether there should be an instance in the system that 
actually makes the distinction $A_1\cap A_2$ through a corresponding third sensor. Note also, that we do not assume that the agen's distinctions imply any 
kind of conscious experiences. \\

In our motivation of $\sigma$-algebras as the right model for describing the intrinsic perspective of an agent we did not explicitly specify the mapping from the world states to internal states. Clearly, this has to be done
based on the formal model of the sensorimotor loop as shown in Figure \ref{sensomot}.
In order to explain how we are going to use $\sigma$-algebras in the context of the sensorimotor loop, consider, for instance, the sensor mechanism $\beta$, which generates a sensor state $s$, given a world state $w$ 
(Note that $\beta$ is meant to incorporate all sensors of the agent, not only one.) 
In general, $s$ will contain information about the world state, which will allow the agent to distinguish it from other world states. One can assign to $\beta$ a set of distinctions, a $\sigma$-algebra, that describes 
the world as seen through the immediate response $s$ of the sensor $\beta$. This $\sigma$-algebra is denoted by $\sigma(\beta)$ and referred to as the $\sigma$-algebra generated by $\beta$. 
Note, however, that $\sigma(\beta)$ does not contain all distinctions that the agent can make. 
There are further distinctions, mediated through the time evolution, 
which incorporates the actuator process of the agent. Therefore, the distinctions that we are going to study are based on both, the sensors and actuators of the agent.      
The intention of the article is to define a 
$\sigma$-algebra of distinctions in the world that describe Uexk\"ull's notion of an agent centric world, the agent's {\em Umwelt\/}.  
In what follows, we address the following two natural problems: 
\begin{enumerate}
\item Which structure in the world is used by the mechanisms of the sensorimotor loop?    
\item Which structure of the world is visible from the intrinsic perspective of the agent?  
\end{enumerate}
We will show that these problems can be appropriately addressed by defining corresponding $\sigma$-algebras. 
The nature of the main results requires some technical knowledge.   
We assume basic knowledge from measure and probability theory and refer to the comprehensive volumes \cite{BogachevI, BogachevII} on measure 
theory and to the textbook \cite{BauerEng} on probability theory. However, the technical part will be complemented by an extended summary and conclusions section on the results and how they relate to Uexk{\"u}ll's work.
The reader interested in the results at a less formal level might want to first read Section \ref{conclusions}.   

\section{Minimal $\sigma$-algebras of the world}\seclabel{separately}
\subsection{Minimal separately measurable $\sigma$-algebra} \label{firstsigma}
We consider measurable spaces $\W$, $\S$, $\C$, and $\A$ as state spaces of the sensorimotor process $W_n, S_n, C_n, A_n$, $n\in {\Bbb N}$. 
For technical reasons, we assume that these are Souslin spaces, equipped with their respective Borel $\sigma$-algebras $\B(\W)$, $\B(\S)$, $\B(\C)$, and $\B(\A)$. 
In the  \apref{quotient}, we collect a few general results that are used in this article, thereby also highlighting the special role of Souslin spaces.  
In order to address the above problems, we fix the Borel \sigalgs\ on the ``agent part''---$\S$, $\C$, and $\A$---of the system. Based on these internal 
$\sigma$-algebras, 
we consider various sub-$\sigma$-algebras of the Borel \sigalg\ on $\W$ that describe  agent related events in the world. 
In measure-theoretic terms,  we study minimal $\sigma$-algebras on $\W$ that satisfy natural measurability conditions. The most natural ansatz is given by the distinctions
that are possible only through sensor measurements. They correspond to the $\sigma$-algebra generated by the kernel $\beta$, that is $\sigma(\beta)$. However, this is not necessarily 
consistent with the world dynamics given by the Markov kernel $\alpha$. Therefore, we consider the following measurability condition.    
We call a
\sigalg\ $\Wsig\subseteq\B(\W)$ \define{(jointly) measurable} if both $\beta$ and $\alpha$ remain measurable
when $\W$ is equipped with $\Wsig$ instead of $\B(\W)$. By general assumption, $\B(\W)$ is jointly measurable.
It turns out that joint measurability is a quite strong condition. Therefore, 
it is natural to consider the following weaker measurability condition. We call $\Wsig$ \define{separately measurable} if,
when we equip $\W$ with $\Wsig$, $\beta$ is measurable and $\alpha$ is separately measurable in the sense that for every
$a\in\A$, the Markov kernel 
\[
    {\alpha_a} : (\W, \Wsig)  \to \P(\W, \Wsig), \qquad w \mapsto \alpha(a, w),
\]
is measurable (``${\mathcal P}$'' denotes the set probability distributions, here on the measurable space $(W,{\mathcal W})$).  Note that, because
$\alpha$ is Borel measurable, the functions $a\mapsto\alpha(a, w)$ are measurable for any $\Wsig\subseteq\B(\W)$.

It is straight-forward to construct the unique minimal (w.r.t.\ partial ordering by inclusion)
separately measurable sub-\sigalg\ $\Wsep$ of $\B(\W)$.

\begin{lemma}\lemlabel{sep}
	Let\/ $\Wsig_0\defeq\sigma(\beta)$ and for\/ $n \in {\Bbb N}$ define\/ $\Wsig_n$ recursively by
		\[ \Wsig_n \defeq \sigma\(\alpha_a\colon \W\to\P(\W, \Wsig_{n-1}),\;a\in\A\) \=
		\sigma\( \kernel{\alpha_a}{\fdot}{B},\;\;a\in\A,\,B\in\Wsig_{n-1} \) \]
	Then\/ $\Wsep\defeq\sigma\(\bigcup_{n \in {\Bbb N}} \Wsig_n\)$ is the unique minimal separately measurable \sigalg, i.e.\
		\[ \Wsep \= \bigcap\bset{\Wsig\subseteq\B(\W)}
			{\text{$\Wsig$ \sigalg, $\beta$ measurable, $\alpha_a$ $\Wsig$-$\Wsig$-measurable for all $a\in\A$}} \]
\end{lemma}
\begin{enproof}
	\item[``$\subseteq$'':] Clearly, any \sigalg\ $\Wsig$ from the set on the right-hand side has to contain
		$\Wsig_0$ because $\beta$ is $\Wsig$\nbd measurable. Further, if it contains $\Wsig_{n-1}$, it
		also has to contain $\Wsig_n$, because $\alpha_a$ must be measurable for all $a\in\A$.
		Thus it contains $\Wsep$.
	\item[``$\supseteq$'':] We have to show that $\Wsep$ is separately measurable. $\beta$ is measurable,
		because $\Wsep \supseteq \Wsig_0=\sigma(\beta)$. $\bigcup_n \Wsig_n$ is an intersection stable
		generator of $\Wsep$, thus for measurability of $\alpha_a$, it is sufficient that
		$\kernel{\alpha_a}{\fdot}{B}$ is $\Wsep$\nbd measurable for $B\in\bigcup_n\Wsig_n$.
		But by definition of $\Wsig_n$, $\kernel{\alpha_a}{\fdot}{B}$ is $\Wsig_n$\nbd measurable for
		$B\in\Wsig_{n-1}$.
\end{enproof}

Note that there is no reason why $\alpha$ should be jointly measurable when we equip $\W$ with $\Wsep$. When we
are working with a separately but not jointly measurable \sigalg, we are rather working with a family
$(\alpha_a)_{a\in\A}$ of kernels than with a single kernel $\alpha$. We do not know if a minimal jointly
measurable \sigalg\ exists in general. The above construction does not work well for $\alpha$ instead of
$\alpha_a$, because we want a product \sigalg\ on $\A\times\W$ and taking products is not compatible with
intersections in the sense that $\Asig\otimes(\Wsig\cap\Wsig') \subsetneqq (\Asig\otimes\Wsig)\cap
(\Asig\otimes\Wsig')$ in general.
Of course, every jointly measurable \sigalg\ is separately measurable and thus has to contain $\Wsep$. Also note
that $\Wsep$ need not be countably generated, which might cause technical problems when working with $\Wsep$.
Next we show that in the ``nice case'' where $\Wsep$ is countably generated, $\alpha$ is jointly measurable.

\begin{proposition}\proplabel{jointmeasurable}
	If\/ $\Wsep$ is countably generated, then it is jointly measurable, and in particular the unique minimal
	jointly measurable \sigalg.
\end{proposition}
\begin{proof}
	Let $\Asig\defeq\B(\A)$. We have to show that $f\colon \A\times\W\to[0,1]$,
	$(a, w)\mapsto\kernel\alpha{a, w}{B}$ is $(\Asig\otimes\Wsep)$\nbd measurable for arbitrary choice of
	$B\in\Wsep$. Because $\Wsep$ is countably generated, $\Asig\otimes\Wsep$ is a countably generated
	sub-\sigalg\ of the Borel \sigalg\ of the Souslin space $\A\times\W$. It follows from Blackwell's
	theorem (see \apref{quotient}) and the fact that $f$ is Borel measurable that $f$ is
	$\Asig\otimes\Wsep$ measurable if and only if it is constant on the atoms of $\Asig\otimes\Wsep$.
	The atoms are obviously of the form $\sset{a}\times F$, where $a\in\A$ and $F$ is an atom of $\Wsep$.
	Because $\alpha_a$ is measurable w.r.t.\ $\Wsep$, $f(a, \fdot)$ is constant on the atom $F$.
	Thus, $f$ is constant on $\sset{a}\times F$ and therefore jointly measurable.
\end{proof}

A simple sufficient condition for $\Wsep$ to be countably generated is that there are only countably many
possible actuator states, i.e.\ $\A$ is countable.

\begin{corollary}\corlabel{acount}
	Let\/ $\A$ be countable. Then\/ $\Wsep$ is countably generated and jointly measurable.
\end{corollary}
\begin{proof}
	Because $\B(\S)$ is countably generated, $\Wsig_0$ is countably generated. If $\Wsig_{n-1}$ is countably
	generated, the same holds for $\sigma\(\alpha_a\colon\W\to\P(\W, \Wsig_{n-1})\)$ for any $a\in\A$.
	Because $\A$ is countable, $\Wsig_n$ is generated by a countable union of countably generated \sigalgs,
	thus it is countably generated, and the same holds for $\Wsep$.
\end{proof}

\subsection[Almost jointly measurable \sigalg\ in the memoryless case]
	{A countably generated, almost jointly measurable \sigalg\ in the memoryless case}\seclabel{almost}

In this section, \emph{we assume that the agent is memoryless, i.e.\ $C_n$ is conditionally independent of
$C_{n-1}$ given $S_n$}. Then we can concatenate the kernels from $\W$ to $\S$, from $\S$ to $\C$, and from $\C$
to $\A$ to obtain a new kernel $\gamma$ from $\W$ to $\A$. We then have the following situation, where the $\C$
and $\S$ components are marginalised (integrated) out.
\renewcommand{\senkern}{{\textstyle\gamma}}
\[ \xymatrix{
	\ar[r]& W_{n-1}\ar[rr]\ar[dr]^\senkern && W_n \ar[rr]\ar[dr]^\senkern&& W_{n+1}\ar[r]\ar[dr]^\senkern & \\
	\ar[ur]^\actkern&& A_{n-1}\ar[ur]^\actkern && A_n\ar[ur]^\actkern && A_{n+1}
} \]
Note that if $\beta$ is $\Wsig$-measurable, the same holds for $\gamma$, but the converse need not be true. We
introduce yet another kernel $\kappa$ which is the combination of $\gamma$ and $\alpha$, i.e.\ $\kappa\colon
\W\to\P(\A\times\W)$, $\kappa(w)\=\gamma(w) \otimes \alpha(\fdot, w)$. The reason to do this is that while
kernels mapping \emph{from} a product space complicate finding minimal \sigalgs\ (with product structure), this
is not the case for kernels mapping \emph{into} a product space. The variables $W_n, A_n$, $n \in {\Bbb N}$, factorise
also according to the following graphical model.
\renewcommand{\senkern}{{\raisebox{0.7ex}{$\kappa$}}}
\[ \xymatrix{
	{\phantom{W_n}} \ar[r]\ar[dr]& W_{n-1}\ar[r]\ar[dr]^\senkern & W_n \ar[r]\ar[dr]^\senkern& W_{n+1}\ar[r]\ar[dr]^\senkern & \\
	& A_{n-2}\ar@{-}[u] & A_{n-1}\ar@{-}[u] & A_n\ar@{-}[u]  &
}\]
We can define the minimal \sigalg\ $\Wam$ such that $\beta$ and $\kappa$ are measurable in the same way as
we defined $\Wsep$.

\begin{lemma}\lemlabel{am}
	Let\/ $\Wsig'_0\defeq\sigma(\beta)$ and
		\[ \Wsig'_n \defeq \sigma\( \kappa \colon \W\to\P(\A\times\W, \Asig\otimes\Wsig'_{n-1})\)
			\= \sigma\(\kernel\kappa\fdot{B},\; B\in\Asig\otimes\Wsig'_{n-1}\).\]
	Then\/ $\Wam\defeq \sigma\(\bigcup_{n\in {\Bbb N}} \Wsig'_n\)$ is the unique minimal \sigalg\ on\/ $\W$ s.t.\
	$\beta$ and\/ $\kappa$ are measurable. Furthermore, $\Wam$ is countably generated.
\end{lemma}
\begin{proof}
	Analogous to the proof of \lemref{sep} and \corref{acount}.
\end{proof}

In the following, consider $\kappa$ as kernel from $(\W, \Wam)$ to $(\A\times \W, \Asig\otimes\Wam)$. Note that
because $\Wam$ is countably generated, the quotient space obtained by identifying atoms of $\Wam$ to points is
again a Souslin space\footnote{more precisely there exists a Souslin topology such that the Borel \sigalg\ coincides
with the final \sigalg\ induced by the canonical projection from $(\W, \Wam)$ onto the quotient} (see
\apref{quotient}). Thus it is technically nice and, in particular, we can factorise $\kappa$ into $\gamma$ and
some kernel $\alpha'$ from $\A\times\W$ to $\W$ by choosing regular versions of conditional probability. Then
$\alpha'$ is jointly measurable. The draw-back is that $\alpha'(\fdot, w)$ is only defined $\gamma(w)$\nbd almost
surely and we cannot guarantee that $\alpha'=\alpha$ is a valid choice, i.e.\ that $\alpha$ is
$\Asig\otimes\Wam$\nbd measurable. We easily get the following.

\begin{lemma}
	$\Wam$ is the unique minimal \sigalg\ on\/ $\W$ that satisfies the following condition. $\beta$ is
	measurable and there exists a (jointly measurable) kernel\/ $\alpha'$ from\/ $\A\times\W$ to\/ $\W$,
	s.t., for every measure\/ $\mu\in\P\(\W, \B(\W)\)$, $\alpha\=\alpha'$ $(\mu\otimes\gamma)$\nbd almost
	surely.
\end{lemma}
\begin{proof}
	The above discussion shows that $\Wam$ satisfies the condition (note that we can w.l.o.g.\ assume that
	$\mu$ is a Dirac measure). If, on the other hand, $\alpha'$ as above exists, then $\kappa$ equals the
	composition of $\gamma$ and $\alpha'$. In particular, $\kappa$ is measurable and \lemref{am} yields the
	claim.
\end{proof}

The condition $\alpha\=\alpha'$ a.s.\ w.r.t.\ every measure of the form $\mu\otimes\gamma$ means that the
difference between $\alpha$ and $\alpha'$ is not visible regardless of any changes we might impose on the
environment. The situation, however, may change if the agent changes its policy $\pi$, thereby changing
the kernel $\gamma$. Then the difference between $\alpha$ and $\alpha'$ can become important and $\alpha'$ as well
as $\Wam$ would have to be changed.

We trivially have that every jointly measurable \sigalg\ $\Wsig$ on $\W$ must contain $\Wam$. In particular, if
$\Wsep$ is countably generated, $\Wam\subseteq\Wsep$. This is probably not true in general.


\section{The world from an intrinsic perspective}
\subsection{Sensory equivalence} \label{sensequiv}
In what follows we use equivalence relations to coarse grain the world states and apply the 
constructions described in \apref{quotient}. 

Denote by $P_S^{a_{\Bbb N}}(w) \in \P(\S^{\Bbb N})$ the distribution of sensor values when the ``world'' is initially in the state
$W_1=w\in\W$ and the agent performs the sequence $a_{\Bbb N} \in\A^{\Bbb N}$, of actuator states, i.e.\ $A_n=a_n$.
That is, we modify the agent policy $\pi$ in a time-dependent way such that $\pi$ is replaced by ${({\pi}_n)}_{n \in {\Bbb N}}$
and $\pi_n$ ignores the memory (and thus the sensors) and outputs (the Dirac measure in) the value $a_n$. The
sensor and world update kernels $\beta$ and $\alpha$, however, remain unchanged. More explicitly, for
$B=B_1\times\cdots\times B_n\times\S\times\cdots\in\B(\S^{\Bbb N})$,
	\[      P_S^{a_{\Bbb N}}(w_1)(B)  \= \kernint{\cdots\;
		\kernint{\kernel\beta{w_1}{B_1}\cdots\kernel\beta{w_n}{B_n}}{\alpha_{a_{n-1}}}{w_{n-1}}{w_n}
		\;\cdots\,}{\alpha_{a_1}}{w_1}{w_2}. \]
Now we define an equivalence relation $\eqsep$, called \emph{sensory equivalence}, on $\W$ by
	\[ w \eqsep w' \defequiv P_S^{a_{\Bbb N}}(w) \=  P_S^{a_{\Bbb N}}(w') \;\; \forall a_{\Bbb N} \in\A^{\Bbb N} .   \]
More generally, we obtain the \emph{intrinsic \sigalg}
	\[ \Wint \defeq \sigma\Bigl( P_S^{a_{\Bbb N}}  ,\; a_{\Bbb N} \in\A^{\Bbb N}\Bigr), \]
which describes the information about the world that can in principle be obtained by the agent through its
sensors. Obviously, the atoms $\ati{\cdot}$ of the intrinsic \sigalg\ are given precisely by the sensory equivalence, i.e.\
$\ati{w} = \set{w'\in\W}{w'\eqsep w}$ for all $w\in\W$.

It turns out (\propref{atoms} and \exref{sepcoarse}) that the intrinsic \sigalg\ leads to a coarser partitioning
of the world than the extrinsic perspective formalised by $\Wsep$. The reason is that the construction of
$\Wsep$ uses knowledge of the mechanisms of the ``world'', more precisely the world update kernel $\alpha$,
which is required to remain measurable when we replace the Borel \sigalg\ $\B(\W)$ by $\Wsep$.  The necessary
information about $\alpha$ cannot be constructed from the sensor values in general, even if infinitely many
observations are possible and all probabilities can be estimated accurately. The difference between the
intrinsic and extrinsic point of view is precisely that the agent does not know the mechanisms of the world
encoded in $\alpha$.

\begin{proposition}\proplabel{atoms}
	$\Wint\subseteq \Wsep$. In particular
		\[ \ats{w} \subseteqs \ati{w} \= \set{w'\in\W}{w'\eqsep w} \qquad\forall w\in\W . \]
	Furthermore, $\Wint=\Wsep$ if and only if\/ $\Wint$ is separately measurable, i.e. $\alpha_a$ is\/
	$\Wint$-$\Wint$-measurable for every\/ $a\in \A$.
\end{proposition}
\begin{proof}
	$\beta$ and $\alpha_a$ are measurable w.r.t.\ $\Wsep$. Because the \sigalg\ on $\P(\S^{\Bbb N})$ is
	generated by the evaluations, and the cylinder sets form a generator of the \sigalg\ on $\S^{\Bbb N}$,
	measurability of $\beta$ and $\alpha_a$ implies measurability of the function $P_S^{a_{\Bbb N}}$ for
	every $a_{\Bbb N} \in\A^{\Bbb N}$. Hence $\Wint \subseteq \Wsep$.
	This directly implies the corresponding inclusion for the atoms.

	If $\Wint=\Wsep$, $\Wint$ is separately measurable by \lemref{sep}. Conversely, assume that $\Wint$ is
	separately measurable. Then, again by \lemref{sep}, $\Wsep\subseteq \Wint$ and hence $\Wint=\Wsep$.
\end{proof}

Equality in the above proposition does not hold in general, as the following example shows.
\begin{example}\exlabel{sepcoarse}
	Let $\W\defeq\sset{1,\ldots,5}$, $\S=\sset{0, 1}$ and $|\A|=1$, i.e.\ the agent is only observing (a
	state-emitting HMM). Let $\beta(1)\=\beta(4)\=\beta(5)\=\half\dirac0+\half\dirac1$,
	$\beta(2)\=\dirac0$, and $\beta(3)\=\dirac1$.  Further let $\alpha(1)\=\half\dirac2+\half\dirac3$,
	$\alpha(2)\=\alpha(3)\=\alpha(4)\=\dirac4$, $\alpha(5)=\dirac5$. $\alpha$ can be illustrated as
	\newcommand{\optxt}[1]{\text{\rlap{$\quad$\small sensor value #1}}}
	\[\xymatrix{
		&        	& 2\ar[dr] & & & \optxt{$0$}\\
		&1\ar[ur]\ar[dr] &          & 4\ar@(u,r) & 5\ar@(u, r) & \optxt{$0$ or $1$}\\
		&	 	& 3\ar[ur] & && \optxt{$1$}
	}\]
	Then $1\eqsep4\eqsep5\neqsep2, 3$ and $\At(\Wsig_0)=\bsset{\{1, 4, 5\}, \{2\}, \{3\}}$, but
		\[\At(\Wsep)\=\At(\Wsig_1)\=\bsset{\{1\},\, \{2\},\, \{3\},\, \{4, 5\}}. \]
	Thus $1$ and $4$ are identified by $\eqsep$ because they produce identical sequences of sensor values,
	but they are not identified by $\Wsep$ because they have non-identified successors. The definition of
	$\Wsep$ requires that $\alpha$ remains unchanged, while the same sensor values can be produced with the
	partition given by $\eqsep$ by modifying $\alpha$ to $\alpha'$, where $\alpha'(1)\=\alpha'(4)$ is an
	arbitrary convex combination of $\alpha(1)$ and $\alpha(4)$.
\end{example}

\subsection{Sensor-preserving modification of the world}

\Exref{sepcoarse} suggests that one might be able to interpret the coarser partition given by sensory
equivalence as description of the relevant part of the world, provided one is allowed to modify the world update
kernel $\alpha$ in such a way that the distribution of sensor values is preserved.
Intuitively, one just has to choose one of the values $\alpha_a$ takes on a given $\eqsep$\nbd equivalence class.

Of course these selections have to be done in a measurable way, and we need technical restrictions to deal with
this problem. Namely, we assume that the world $\W$ is \emph{compact}, and the sensor kernel $\beta$ as well as
the world update kernels $\alpha_a\colon \W\to\P(\W)$ for every given action $a\in\A$ are \emph{continuous}.
As usual, $\P(\W)$ is equipped with the weak topology induced by bounded continuous functions. Note that
compactness and metrisability of $\W$ also implies compactness and metrisability of $\P(\W)$.

Under these assumptions we can prove that it is possible to modify the world update (and with it the smallest
separately measurable \sigalg\ $\Wsep$) in such a way that the sensor process is preserved and equality holds in
\propref{atoms}. Furthermore, the ``new $\Wsep$'' is countably generated and jointly measurable for the modified
system.

\begin{definition}
	Let\/ $\alpha'\colon \A\times\W\to\P(\W)$ be an ``alternative'' world update kernel. 
	\begin{enumerate}
	\item We call $\alpha'$ \define{equivalent} to $\alpha$ if for every $w\in\W$ and $a_{\Bbb N} \in\A^{\Bbb N}$ the sensor
		process $P_S^{a_{\Bbb N}}(w)$ coincides with the sensor process $P_{S, \alpha'}^{a_{\Bbb N}} (w)$ obtained by replacing $\alpha$ with
		$\alpha'$.
	\item Denote by $\Wsep^{\alpha'}$ the smallest separately measurable \sigalg\ of the system where $\alpha$ is
		replaced by $\alpha'$.
	\end{enumerate}
\end{definition}

\begin{remark}
	$\alpha'$ can be seen as a \emph{model} for the mechanisms of the world, which the agent might use. If
	$\alpha'$ is equivalent to $\alpha$, $\alpha'$ is a perfect model, as far as the agent's
	(possible) observations are concerned. Of course it can still make wrong assumptions about aspects of
	the world that cannot be inferred by the agent. In \propref{compact}, we show (under a continuity
	and compactness assumption) that the agent can always build a perfect model in this sense which is
	consistent with his intrinsic \sigalg.
\end{remark}

\begin{lemma}\lemlabel{PSacont}
	Assume that\/ $\alpha_a$ is continuous for every\/ $a\in\A$ and\/ $\beta$ is continuous.
	Then\/ $P_S^{a_{\Bbb N}}$ is continuous for every $a_{\Bbb N} \in\A^{\Bbb N}$.
\end{lemma}
\begin{proof}
	 Easy to see directly or a special case of \cite[Thm.~1]{Karr}.
\end{proof}

\begin{proposition}\proplabel{compact}
	Let\/ $\W$ be compact, $\beta$ and\/ $\alpha_a$ continuous for every\/ $a\in \A$. Then there is a
	kernel\/ $\alpha'\colon\A\times\W\to\P(\W)$, such that\/ $\alpha'$ is equivalent to\/ $\alpha$ and\/
	\begin{equation}\eqlabel{alpha'}
		\Wsep^{\alpha'} \= \Wint.
	\end{equation}
	In particular,
		\[ {[w]}_{\Wsep^{\alpha'}} \= \set{w'\in\W}{w'\eqsep w} \qquad\forall\, w\in\W. \]
	Furthermore, $\Wsep^{\alpha'}$ is countably generated as well as
	jointly measurable (for the modified system with $\alpha$ replaced by $\alpha'$).
\end{proposition}
\begin{enproof}
	\item Let $X \defeq \P(\S^{\Bbb N})^{\A^{\Bbb N}}$ be the set of mappings from action sequences to
		distributions of sensor sequences, equipped with product topology. Given an initial state $w$ of
		the world, denote by $F(w)$ the corresponding kernel from action sequences to sensor sequences,
		i.e.\ $F\colon\W\to X$, $F(w) \= \(a_{\Bbb N} \mapsto P_S^{a_{\Bbb N}}(w) \)$. Note that $F$
		generates $\Wint$, i.e.\ $\sigma(F)=\Wint$.
		Because every $P_S^{a_{\Bbb N}}$ is continuous (\lemref{PSacont}) and $X$ carries the product
		topology, $F$ is a continuous function from the compact metrisable space $\W$ into the Hausdorff
		space $X$. In particular, the image $F(\W)$ is also compact and metrisable.
	\item We can apply a classical selection theorem, e.g.\ Theorem~6.9.7 in \cite{BogachevII}, and obtain a
		measurable right-inverse $G\colon F(\W)\to \W$ with $F\circ G \= \id_{F(\W)}$.
		Define $\varsigma \defeq G\circ F$. Then $\varsigma$ is measurable, and $\sigma(\varsigma)
		\subseteq \sigma(F)$. On the other hand, $\sigma(F) = \sigma(F\circ G \circ F) \subseteq
		\sigma(\varsigma)$. Hence,
		\begin{equation}\eqlabel{atoms}
			\sigma(\varsigma)  \= \sigma(F) \= \Wint.
		\end{equation}
		Define $\alpha'_a \defeq \alpha_a\circ\varsigma$ for every $a\in \A$.
	\item A simple induction shows that $\alpha'$ is indeed equivalent to $\alpha$:
		For $B\= B_1\times\cdots \times B_n\times \S \times \cdots \in \B(\S^{\Bbb N})$
		and $C\defeq B_2\times \cdots \times B_n \times \S\times \cdots$, we obtain by induction over
		$n$
		\[ 
		P_{S, \alpha'}^{a_{\Bbb N}} (w)(B) \= \int \beta(w ; B_1)
				P_{S, \alpha'}^{a_{\{2,3,\dots\}}} (\cdot) (C) \, {\rm d} \alpha_{a_1}'(w) 
			\= P_S^{a_{\Bbb N}}\(\varsigma(w)\)(B) \= P_S^{a_{\Bbb N}} ( w )(B)
		\]
	\item We claim that $\Wsep^{\alpha'} \= \sigsig$. Indeed, $\alpha'_a$ is $\sigsig$\nbd measurable
		by definition, and $\beta$ is $\sigsig$\nbd measurable, because $\beta(w)$ is a marginal of
		$F(w)(a_{\Bbb N})$ for any $a_{\Bbb N}$. Therefore, $\sigsig$ is separately
		measurable in the modified system and $\Wsep^{\alpha'}   \subseteq \sigsig$.
		On the other hand, $P_S^{a_{\Bbb N}} = P_{S, \alpha'}^{a_{\Bbb N}}$ is $\Wsep^{\alpha'}$\nbd measurable, thus the same holds for $F$ and
		$\sigsig=\sigma(F) \subseteq \Wsep^{\alpha'}  $. Hence $\Wint=\Wsep^{\alpha'}$ follows from
		\eqref{atoms}.
	\item Since $\varsigma$ is a function into a space with countably generated \sigalg,
		$\Wsep^{\alpha'} \=\sigma(\varsigma)$ is countably generated. In particular, it is jointly measurable by
		\propref{jointmeasurable}.
\end{enproof}

\begin{remark}[Non-compact $\W$]
	For a non-compact world $\W$, we can still obtain an equivalent $\alpha'$ satisfying \eqref{alpha'} if
	we relax the condition that $\alpha'$ needs to be Borel measurable. Instead, it is only universally
	measurable, i.e.\  $\mu$\nbd measurable for every $\mu\in\P(\W)$. To see this, just replace the
	selection theorem used in the proof of \propref{compact} by a selection theorem for Souslin spaces,
	e.g.\ Theorem~6.9.1 in \cite{BogachevII}. The drawback is that the universal \sigalg\ is not countably
	generated and we do not obtain joint measurability of $\Wsep^{\alpha'}$.
\end{remark}


\section{Summary and conclusions} 
\label{conclusions}
\subsection{Summary of our definitions and results}
We started with the mathematical description of the agent's interaction with the world in terms of a causal diagram. This lead us  
to the definition of the agent's sensorimotor loop (see Figure \ref{sensomot}), which formalises Uexk{\"u}ll's fundamental 
notion of a {\em function-circle\/}.       
The sensorimotor loop contains, as part of the description, 
a reference world, referred to as the {\em outer world\/} by Uexk{\"u}ll. It is considered to be objective in the sense that it sets constraints on the distinctions that {\em any\/} observer can make in that world. 
This is formalised in terms of a ``large'' $\sigma$-algebra which contains all reasonable distinctions (the Borel $\sigma$-algebra of the world).  
We defined  
sub-$\sigma$-algebras ${\mathcal W}_{\rm ext}$ and ${\mathcal W}_{\rm int}$ that represent two agent specific perspectives. 
The first one, 
introduced in Section \ref{firstsigma}, is  
based on two requirements:
\begin{enumerate}
\item First, we assume that ${\mathcal W}_{\rm ext}$ contains the distinctions in the world that the agent can make based on the immediate response of its sensors. 
As the corresponding mechanism is encoded by the Markov kernel $\beta$, this means that the $\sigma$-algebra
generated by $\beta$ should be contained in ${\mathcal W}_{\rm ext}$, that is $\sigma(\beta) \subseteq {\mathcal W}_{\rm ext}$ (see \lemref{sep}).
These distinctions seem to be closely related to
Uexk{\"u}ll's {\em world-as-sensed\/} (translation of the original term {\em enfache Merkwelt\/} \cite{Uexkull26}, page 132). In addition to these aspects of the world, the agent is capable of making also 
mediated or distal distinctions. We believe that the mediated distinctions correspond to those aspects of the world that 
Uexk{\"u}ll describes as {\em the higher grades of the world-as-sensed\/} (translation of the original term {\em h\"ohere Stufen der Merkwelt\/} \cite{Uexkull26}, page 140).
In order to incorporate these mediated distinctions, we impose the next condition on ${\mathcal W}_{\rm ext}$.         
\item With this second condition we basically assume that the system is closed in the sense that all mediated distinctions in the world are taken into account. This is formalised by the iteration formula for ${\mathcal W}_n$ in  
\lemref{sep}. The main insight of this lemma is that the closedness is equivalently expressed by the invariance condition 
\begin{equation}  \label{invariance}
    \alpha_a^{-1}({\mathcal W}) \, \subseteq \, {\mathcal W} , \qquad a \in \A .
\end{equation}
Stated differently, incorporating all mediated distinctions, based on an initial set of prime distinctions, is equivalent to enlarging this set until the condition (\ref{invariance}) is satisfied.   
Our $\sigma$-algebra ${\mathcal W}_{\rm ext}$ is then the smallest $\sigma$-algebra ${\mathcal W}$ that contains $\sigma(\beta)$ 
and satisfies this invariance.  Condition (\ref{invariance}) is related to the closedness of dynamical systems studied in \cite{PfanteAy, comparison}.
\end{enumerate} 
The definition of ${\mathcal W}_{\rm ext}$ is quite natural and might appear as the right formalisation of Uexk{\"u}ll's {\em Umwelt\/} in terms of a $\sigma$-algebra of distinctions.  
However, the invariance condition requires knowledge about the mechanisms of the world, formalised in terms of the Markov kernel $\alpha$. 
Therefore, in Section \ref{sensequiv} we introduced another $\sigma$-algebra, which does not require this knowledge and is defined in an intrinsic manner. It is based on the following sensory equivalence relation: 
We identify two world states $w$ and $w'$ if they induce
the same sensor process, given any sequence $a_1,a_2,\dots$ of actuator states. 
We define ${\mathcal W}_{\rm int}$ to be the $\sigma$-algebra associated with this relation. It consists of those distinctions the agent can make in the world that  involve both the sensors and actuators. In this sense, 
${\mathcal W}_{\rm int}$ takes into account Uexk{\"u}ll's {\em perceptual world\/} ({\em Merkwelt\/}, referred to as {\em world-as-sensed\/} in  \cite{Uexkull26}) and {\em effector world\/} ({\em Wirkwelt\/}, referred to as 
{\em world of action\/} in  \cite{Uexkull26}). We do not, however, separate these two worlds as they are intertwined and define, together,  ${\mathcal W}_{\rm int}$.

\begin{quote}
``We no longer regard animals as mere machines, but as subjects whose essential activity consists of perceiving and acting. We thus unlock the gates that lead to other 
realms, for all that a subject perceives becomes his {\em perceptual world\/} and all that he does, his {\em effector world\/}. Perceptual and effector worlds together form a closed unit, the
{\em Umwelt\/}.'' (\cite{Uexkull34}, page 320)  
\end{quote}

In Section \ref{separa} below,
we comment on other ways to combine 
the {\em perceptual world\/} and the {\em effector world\/}. But first, let us address the following question: 
How does ${\mathcal W}_{\rm int}$ relate to ${\mathcal W}_{\rm ext}$? Their relation is actually quite interesting. 
First of all, as one would like to have, ${\mathcal W}_{\rm int}  \subseteq {\mathcal W}_{\rm ext}$, which is the main content of \propref{atoms}.
This means that one can attribute more distinctions to the agent, if the mechanism $\alpha$ of the world is known. Or, stated differently, the 
agent can operate on the basis of distinctions that are not internally indentifiable in terms of its sensors and actuators. 
Note that there is one apparent limitation of ${\mathcal W}_{\rm int}$: It is not invariant in the sense of (\ref{invariance}), which means that it does not describe a closed system. 
On the other hand, making it invariant by extending it sufficiently already leads to our previous $\sigma$-algebra ${\mathcal W}_{\rm ext}$. 
 However, the violation of (\ref{invariance}) can only be seen from outside. The agent has no access to the mechanism $\alpha$ from its intrinsic perspective, and it simply does not see whether or not (\ref{invariance})
 is statisfied. It is quite surprising that, according to our \propref{compact}, 
 it is possible for the agent to imagine a mechanism $\alpha'$ that is compatible with ${\mathcal W}_{\rm int}$ in the sense that it satisfies (\ref{invariance}) where 
 $\alpha$ is replaced by $\alpha'$. With this modification of the mechanism, we have ${\mathcal W}_{\rm int} = {\mathcal W}^{\alpha'}_{\rm ext}$, where ${\mathcal W}^{\alpha'}_{\rm ext}$ is defined in the same 
 way as ${\mathcal W}_{\rm ext}$ but with the mechanism $\alpha'$ instead of $\alpha$. Even if the agent is actually operating on the basis of distinctions 
 that are not identifiable from its intrinsic perspective, it is always possible to imagine different mechanisms that only involve the identifiable distinctions.
This is why we think that ${\mathcal W}_{\rm int}$ is the right object for describing the {\em Umwelt\/} of an agent.     

\subsection{On decompositions of the {\em Umwelt\/} into {\em Merkwelt\/} and {\em Wirkwelt}} \label{separa}
Our approach starts with a set of prime distinctions and extends this set by taking into account mediated distinctions that are generated in terms of the actuators of the agent.   
This way, the {\em perceptual world\/} ({\em Merkwelt\/}) and the {\em effector world\/} ({\em Wirkwelt\/}) are incorporated into the {\em Umwelt\/} in an asymmetric manner. 
One could also try to define both worlds separately and integrate them in a symmetric way. Our impression is, however, that this approach has limitations, which we are going to 
briefly explain. \\

As argued above, ${\mathcal W}_{\rm merk} := \sigma(\beta)$ already models the 
{\em world-as-sensed\/}, which we consider, for the moment, to be the same as the {\em perceptual world\/}, as both terms are translations of {\em Merkwelt\/} (note, however, that {\em world-as-sensed\/} denotes its 
simple form,  {\em einfache Merkwelt\/}). Constructing a corresponding {\em effector world\/} as $\sigma$-algebra on $\W$ appears less natural. To see this, compare in  Figure \ref{sensomot} 
the causal link from the world state $W_n$ to the sensor state $S_n$ with the 
causal causal link from the actuator state $A_{n-1}$ to the world state $W_{n}$. One is directed away from and one toward $W_n$. Generally, it is natural to ``pull back'' distinctions, as we did for the sensor kernel. 
But without further assumptions there is no natural way to ``push forward'' distinctions, which would be required for the definition of an {\em effector world\/} as $\sigma$-algebra on $\W$. 
Instead of specifying these assumptions and discussing related technical problems, let    
us simply assume that we already have both, the {\em perceptual world\/} 
${\mathcal W}_{\rm merk}$ and the {\em effector world\/} ${\mathcal W}_{\rm wirk}$. 
How would one combine them so that they ``form a closed unit''
as Uexk{\"u}ll describes it in the above quote? There are two natural ways to combine the two $\sigma$-algebras: The intersection ${\mathcal W}_\cap := {\mathcal W}_{\rm merk} \cap {\mathcal W}_{\rm wirk}$, which consists of all 
distinctions, that are contained in both worlds, and the union, ${\mathcal W}_\vee := \sigma(  {\mathcal W}_{\rm merk} \cup {\mathcal W}_{\rm wirk} )$, which is the smallest set of distinctions so that both worlds are contained in it. 
We argue that both choices are limited by applying them to two
special cases of the senserimotor loop: In case one, we assume that the agent can only observe but not act, which we refer to as a {\em passive observer\/}. Any reasonable definition 
of the $\sigma$-algebra ${\mathcal W}_{\rm wirk}$ 
should lead to a trivial 
{\em effector world\/} in this case, that  is ${\mathcal W}_{\rm wirk} = \{\emptyset, \W\}$. In case two,
we assume that the agent can act but has no sensors to perceive the consequences of its actions. We refer to this second agent as a {\em blind actor\/}. For the blind actor, we obviously 
have a trivial {\em perceptual world\/}, that is ${\mathcal W}_{\rm merk} = \{\emptyset, \W\}$.
The following 
table summarises the various resulting {\em Umwelten\/} for these two cases: 

\begin{center}
\begin{tabular}{ c || c  | c  | c } 
                                   & ${\mathcal W}_{\cap}$ &  ${\mathcal W}_{\vee}$ &  ${\mathcal W}_{\rm int}$ \\ \hline \hline
  passive observer  & no {\em Umwelt\/}  & equals  {\em Merkwelt\/}  &  contains {\em Merkwelt\/}     \\
  blind actor              & no {\em Umwelt\/}  & equals {\em Wirkwelt\/}    & no {\em Umwelt\/}         \\ 
\end{tabular}
\end{center}

Let us discuss and compare these simple but instructive outcomes. First, we see that for both, the passive observer and the blind actor, the intersection
${\mathcal W}_{\cap}$ is the trivial $\sigma$-algebra $\{\emptyset, \W \}$, which we interpret as a trivial or no {\em Umwelt\/}. 
We argue, however, that one should attribute a non-trivial {\em Umwelt\/} to a passive observer, if the observed world is rich enough. The reason is that the agent is capable of making distinctions based on its sensors,  
without the involvement of its actuators.   
How about the union ${\mathcal W}_\vee$? In this case, we argue that the outcome for the blind actor is not satisfying. A blind actor 
generates effects in the outer world which are, in principle, visible from the perspective of an external observer in terms of his own set of distinctions. 
But the blind actor himself has no instantiation of these distinctions made from outside. 
Would one still attribute these distinctions to the {\em Umwelt\/} of the blind actor? We think that this should not be the case.   
Finally, as already stated, our definition ${\mathcal W}_{\rm int}$ treats the two worlds in an asymmetric way. This asymmetry is also expressed by the fact that 
the passive observer has a possibly non-trivial {\em Umwelt\/} (the distinctions are made by the passive observer himself), whereas the blind actor   
has a trivial one (the effects are visible only in terms of distinctions made by an external observer).
Our way of integrating sensing and acting in conceptually different from the way sketched in this section and can be summarised as follows: 
The prime object is the {\em perceptual world\/}! Starting with an initial {\em perceptual world\/} (world-as-sensed), the agent can generate more and more distinctions by utilising its actuators. 
Thereby, the {\em perceptual world\/} is gradually enlarged until it incorporates all distinctions of the agent's {\em Umwelt\/} ({\em higher grades of the world-as-sensed\/}).

\subsection{How to treat the case of multiple agents}
We conclude with a rough description of how one can study multiple agents based on the developed tools of this article.  
For this, we have to couple the individual sensorimotor loops, which Uexk{\"u}ll beautifully describes as follows: 

\begin{quote}
``The function-circles of the various animals connect up with one another in the most various ways, and together form the function-world of living organisms, within which plants are included. 
For each individual animal, however, its function-circles constitute a world by themselves, within which it leads its existence in complete isolation.'' (\cite{Uexkull26}, page 126) 
\end{quote}

We sketch 
our ideas on this subject by considering the case of only two agents. 
The diagram in Figure \ref{sensomottwo} shows how the sensorimotor loops of two agents are intertwined. 
Note that each agent $i$ has its own mechanisms $\beta^{(i)}$, $\varphi^{(i)}$, $\pi^{(i)}$, except that there is only one mechanism $\alpha$ which governs the transitions of the common world given the actuator states of both agents. 
Let us first take the perspective of agent one. From this perspective, the outer world includes agent two, that is the variable $W_n^{(1)}$ contains $W_n$, $S^{(2)}_n$, $C^{(2)}_n$, and $A^{(2)}_n$ 
(see Figure \ref{SML_6} (a)). The perspective of agent two is symmetric. Here, the outer world $W_n^{(2)}$ of agent two includes the corresponding variables     
$W_n$, $S^{(1)}_n$, $C^{(1)}_n$, and $A^{(1)}_n$ (see Figure \ref{SML_6} (b)). Obviously, in principle the two outer worlds are not contained in each other and, in particular, they are not identical. Furthermore, they share the process $W$ 
which can be considered as a common world of the two agents. \\ 

\renewcommand{\actkern}{{\qquad\qquad\raisebox{5ex}{$\alpha$}}}
\newcommand{\intkerna}{{\qquad\qquad\raisebox{1ex}{$\varphi^{(1)}$}}}
\newcommand{\senkerna}{\textstyle \beta^{(1)}}
\newcommand{\polkerna}{{\textstyle\pi}^{(1)}}
\newcommand{\intkernb}{{\qquad\qquad\raisebox{-3ex}{$\varphi^{(2)}$}}}
\newcommand{\senkernb}{{\textstyle\beta^{(2)}}}
\newcommand{\polkernb}{{\textstyle\pi}^{(2)}}
\begin{figure}[h]
\begin{center}        
\[ \xymatrix{
          \ar[rr]_\intkernb        &       & C^{(2)}_{n-1} \ar[rr]_\intkernb \ar[d]^\polkernb  &                               & C^{(2)}_n   \ar[rr]_\intkernb \ar[d]^\polkernb &    &  C^{(2)}_{n+1}  \ar[d]^\polkernb  \\
           \ar[dr]    &                       S^{(2)}_{n-1}  \ar[ur]       & A^{(2)}_{n - 1} \ar[dr]  &                S^{(2)}_n  \ar[ur]       & A^{(2)}_{n} \ar[dr] &  S^{(2)}_{n+1}     \ar[ur]  &  A^{(2)}_{n + 1} \\
	\ar[r]&	W_{n-1}\ar[rr]\ar[d]^\senkerna\ar[u]_\senkernb & & W_n\ar[rr]\ar[d]^\senkerna\ar[u]_\senkernb & & W_{n+1}\ar@{-}[r]\ar[d]^\senkerna\ar[u]_\senkernb &\\
	\ar[ur]^\actkern & S^{(1)}_{n-1}\ar[dr] & A^{(1)}_{n-1} \ar[ur]^\actkern & S^{(1)}_n\ar[dr] & A^{(1)}_n\ar[ur]^\actkern & S^{(1)}_{n+1}\ar[dr]&A^{(1)}_{n+1}\\
	\ar[rr]^\intkerna  &&{C^{(1)}_{n-1}} \ar[rr]^\intkerna \ar[u]_\polkerna & & {C^{(1)}_n} \ar[rr]^\intkerna \ar[u]_\polkerna & & {C^{(1)}_{n+1}}\ar[u]_\polkerna
} \]
\end{center}
\caption{The sensorimotor loops of two agents.}
\label{sensomottwo}
\end{figure}
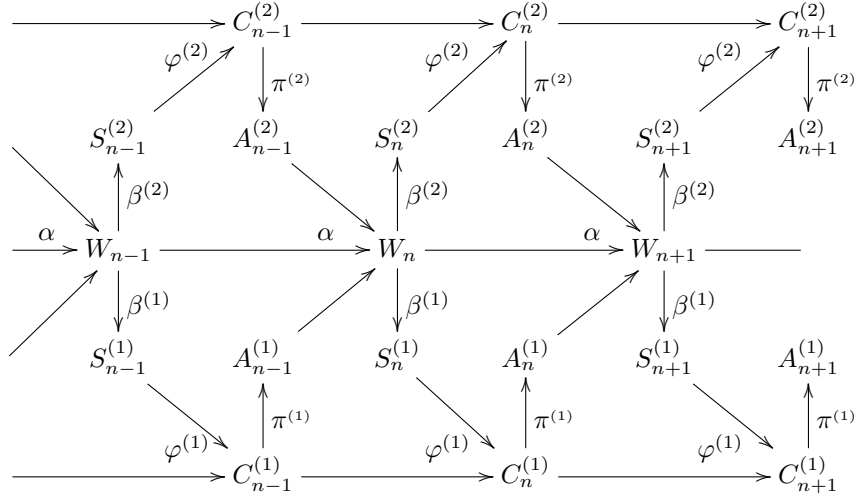

\begin{figure}[h]
\begin{center}        
          \includegraphics[width=15cm]{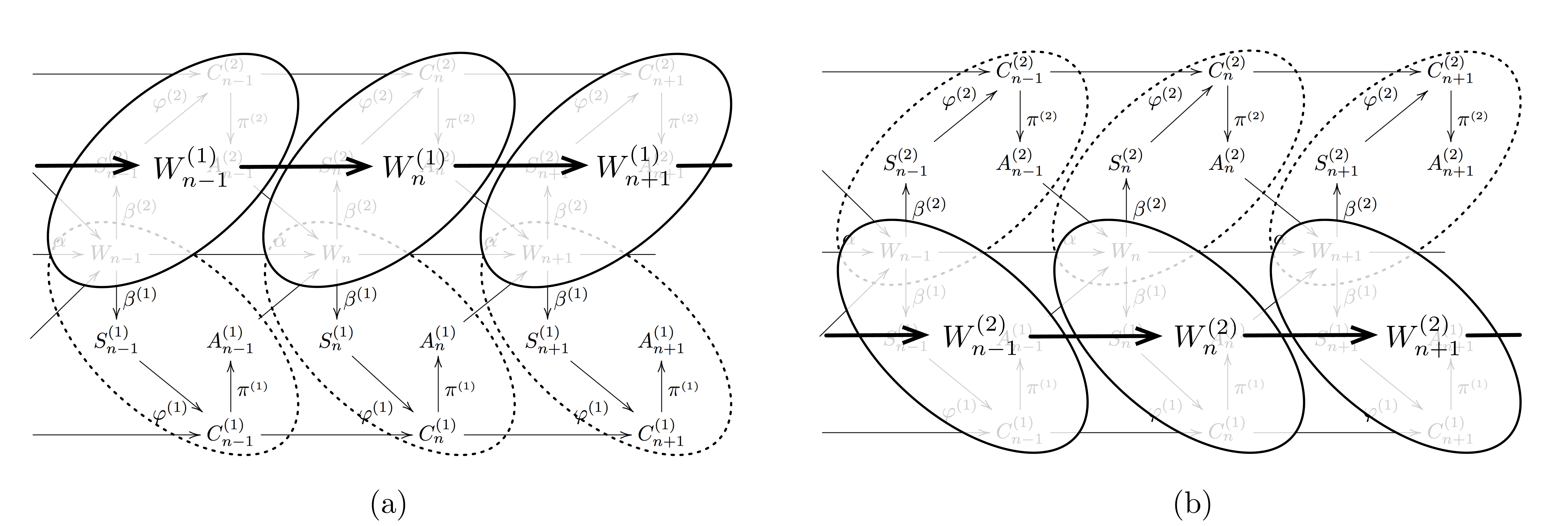} 
                \vspace{-8mm}
\end{center}
\caption{Overlapping but distinct outer worlds of two agents. (a) Outer world of agent one which includes inner world of agent two, and (b) outer world of agent two which includes inner world of agent one.}
\label{SML_6}
\end{figure}

\noindent
We can express the fact that the two outer worlds are different more formally by 
\begin{eqnarray}
    \mathfrak{B}(\W^{(1)}) & = & \mathfrak{B}(\W) \otimes \mathfrak{B}(\S^{(2)}) \otimes \mathfrak{B}(\C^{(2)}) \otimes \mathfrak{B}(\A^{(2)}) \, , \\ 
    \mathfrak{B}(\W^{(2)}) & = & \mathfrak{B}(\W) \otimes \mathfrak{B}(\S^{(1)}) \otimes \mathfrak{B}(\C^{(1)}) \otimes \mathfrak{B}(\A^{(1)})  \, ,
\end{eqnarray}       
and we have ${\mathcal W}^{(i)}_{\rm int} \subseteq {\mathcal W}^{(i)}_{\rm ext} \subseteq  \mathfrak{B}(\W^{(i)})$, $i = 1,2$. The intersections of these sets of distinctions will be contained in the common world, that is
\begin{equation}
    {\mathcal W}^{(1)}_{\rm int} \cap {\mathcal W}^{(2)}_{\rm int} \; \subseteq \; {\mathcal W}^{(1)}_{\rm ext} \cap {\mathcal W}^{(2)}_{\rm ext} \; \subseteq \;  \mathfrak{B}(\W) \, .
\end{equation}
This highlights an important point. The two agents can share some distinctions in the world. However, these have to be contained in the common world. Each agent can make further distinctions in its respective 
outer world, so that the {\em Umwelten\/} will be in general different. 
If we consider a probability measure, however, then we can actually identify distinctions $A_1 \in  {\mathcal W}^{(1)}_{\rm int} $ of agent one with distinctions 
$A_2\in  {\mathcal W}^{(2)}_{\rm int}$ of agent two if their intersection $A_1 \cap A_2$ has full probability. 
This way, the two agents can, in principle, synchronise and reach some consensus on their respective intrinsic worlds.
The more generic situation will be, that the intrinsic worlds are similar rather than perfectly 
identical with respect to the underlying probability measure. In order to quantify how close the intrinsic worlds ${\mathcal W}^{(i)}_{\rm int}$, the {\em Umwelten\/}, of individual agents are, 
appropriate distance measures for $\sigma$-algebras 
will be required. Such measures have been studied in the probability theory and statistics literature \cite{Boylan71, Neveu72, Rogge74}, and might be applicable to the present context.
However, it is not within the scope of this article to present and discuss these measures.    

\section*{Acknowledgement} 
Nihat Ay is grateful for stimulating discussions with Keyan Ghazi-Zahedi and Guido Mont{\'u}far.
The authors would like to thank J\"urgen Jost for his helpful comments. 
\bibliography{bib-math}

\appendix
\section{Appendix: state reduction and quotient construction}\seclabel{quotient}

Let $(X, \F)$ be a measurable space. The \define{atom} of $\F$ containing $x\in X$ and the set of atoms of $\F$
are defined as
	\[ \at{x} \defeq \bigcap_{x\in F\in\F} F \und \At(\F) \defeq \bset{\at{x}}{x\in X}. \]
Note that if $\F$ is countably generated, $\at{x}$ is a measurable set, $\at{x}\in\F$. In general, however,
$\at{x}$ need not be measurable. We recall Blackwell's theorem.

\theoremstyle{plain}\newtheorem*{blackwell}{Blackwell's theorem}
\begin{blackwell}
	Let\/ $X$ be a Souslin space and\/ $\F\subseteq\B(X)$ a countably generated sub-\sigalg\ of the Borel\/
	\sigalg. Then
		\[ \F \= \Bset{F\in\B(X)}{F=\bigcup_{x\in F} \at{x}}. \]
\end{blackwell}

\begin{corollary}
	Let\/ $X$ be a Souslin space, $\F\subseteq\B(X)$ a countably generated\/ \sigalg, and\/ $f\colon X\to\R$
	measurable. Then\/ $f$ is\/ $\F$\nbd measurable if and only if it is constant on the atoms of\/ $\F$.
\end{corollary}

According to Blackwell's theorem, a countably generated sub-\sigalg\ of a Souslin space $X$ is uniquely
determined by the set of it atoms.  $\At(\F)$ is a partition of $X$ into $\B(X)$-measurable sets. Note, however,
that not every partition of $X$ into measurable sets is the set of atoms of a countably generated sub-\sigalg\
of $\B(X)$.

Given any measurable space $(X, \F)$, we can define the quotient space $X_\F$ as the set $\At(\F)$ of atoms of
$\F$ equipped with the final \sigalg\ $\XF$ of the canonical projection $\at{\fdot}\colon X\to\At(\F)$. Then a
set $B\subseteq X_\F$ of atoms is by definition measurable iff $\bigcup B=\bigcup_{\at{x}\in B} \at{x} \in \F$.
Note that, obviously, $B\mapsto\bigcup B$ is a complete isomorphism of boolean algebras from $\XF$ onto $\F$.
The following lemma follows easily from the standard theory of analytic measurable spaces and is one of the
reasons why Souslin spaces, rather than Polish spaces, are the ``right'' class of spaces to work with in our
setting.

\begin{lemma}
	Let\/ $X$ be a Souslin space and\/ $\F\subseteq\B(X)$ a sub-\sigalg. Then\/ $\XF$ is the Borel\/
	\sigalg\ of some Souslin topology on\/ $X_\F$ if and only if\/ $\F$ is countably generated.
\end{lemma}

\end{document}